\documentclass{amsart}
\usepackage{amssymb}
\usepackage{enumerate}
\usepackage{amsmath}
\usepackage{mathrsfs}
\usepackage{tikz}
\usepackage{float}
\copyrightinfo{2015}{H. De Bie \emph{et al.}}

\newtheorem{proposition}{Proposition}
\newtheorem{lemma}{Lemma}

\theoremstyle{definition}

\theoremstyle{remark}
\newtheorem{remark}{Remark}

\begin{document}
\title[A higher rank Racah algebra and the $\mathbb{Z}_2^{n}$ Laplace-Dunkl operator]{A higher rank Racah algebra \\and the $\mathbb{Z}_2^{n}$ Laplace-Dunkl operator}
\author[H. De Bie]{Hendrik De Bie}
\address{Department of Mathematical Analysis, Faculty of Engineering and Architecture, Ghent University, Galglaan 2, 9000 Ghent, Belgium}
\email{Hendrik.DeBie@UGent.be}
\author[V.X. Genest]{Vincent X. Genest}
\address{Department of Mathematics, Massachusetts Institute of Technology, 77 Massachusetts Ave, Cambridge, MA 02139, USA}
\email{vxgenest@mit.edu}
\author[W. van de Vijver]{Wouter van de Vijver}
\address{Department of Mathematical Analysis, Faculty of Engineering and Architecture, Ghent University, Galglaan 2, 9000 Ghent, Belgium}
\email{Wouter.vandeVijver@UGent.be}
\author[L. Vinet]{Luc Vinet}
\address{Centre de Recherches Math\'ematiques, Universit\'e de Montr\'eal, P.O. Box 6128, Centre-ville Station, Montr\'eal, QC H3C 3J7, Canada}
\email{vinet@crm.umontreal.ca}
\subjclass[2010]{81Q80, 81R10, 81R12}
\date{}
\dedicatory{}
\begin{abstract}
A higher rank generalization of the (rank one) Racah algebra is obtained as the symmetry algebra of the Laplace-Dunkl operator associated to the $\mathbb{Z}_2^n$ root system. This algebra is also the invariance algebra of the generic superintegrable model on the $n$-sphere. Bases of Dunkl harmonics are constructed explicitly using a Cauchy-Kovalevskaia theorem. These bases consist of joint eigenfunctions of labelling Abelian subalgebras of the higher rank Racah algebra. A method to obtain expressions for both the connection coefficients between these bases and the action of the symmetries on these bases is presented.
\end{abstract}
\maketitle


\section{Introduction}
The purpose of this paper is to introduce a higher rank generalization of the Racah algebra. This algebra will emerge in the examination of the multifold tensor product representations of $\mathfrak{su}(1,1)$ associated to the $\mathbb{Z}_2^{n}$ Laplace-Dunkl operator. The higher rank Racah algebra will also be shown to be the symmetry algebra of the generic (scalar) superintegrable model on the $n$-sphere. 

The (rank 1) Racah algebra was initially introduced in \cite{Granovskii&Zhedanov-1988}  to provide an underpinning for the symmetries of the $6j$-symbols or Racah coefficients of $SU(2)$. It is defined as the infinite-dimensional associative algebra over $\mathbb{C}$ with generators $K_1$ and $K_2$ that satisfy, together with their commutator $K_3=\left[K_1,K_2\right]$, the following relations:
\begin{align}\label{eq:Racah}
\begin{split}
\left[K_2,K_3\right]&=K_2^2+\left\{K_1,K_2\right\}+dK_2+e_1, \\
\left[K_3,K_1\right]&=K_1^2+\left\{K_1,K_2\right\}+dK_1+e_2,
\end{split}
\end{align}
with $\left\{A,B\right\}:=AB+BA$ and where $d$, $e_1$, $e_2$ are structure constants. This algebra corresponds to the $q=1$ case of the Askey--Wilson or Zhedanov algebra \cite{Zhedanov-1991}, which encodes the bispectral properties of the Askey--Wilson polynomials, sitting atop the Askey scheme of hypergeometric orthogonal polynomials  \cite{Koekoek&Lesky&Swarttouw-2010}. In a similar fashion, the Racah algebra encodes the bispectral properties of the Racah polynomials. The Racah algebra has an alternative presentation in terms of generators $L_1$, $L_2$, $L_3$ and $F$ which is commonly referred to as the ``equitable'' or ``democratic'' presentation \cite{Levy-Leblond&Levy-Nahas-1965, Terwilliger-2006}. In this form one has $L_1+L_2+L_3=G$ and $\left[L_1,L_2\right]=\left[L_2,L_3\right]=\left[L_3,L_1\right]=F$, where $G$ is a central operator. The commutation relations then read 
\begin{align} \label{eq:eqRacah}
\begin{split}
\left[L_1,F\right]&=L_2L_1-L_1L_3+i_1, \\
\left[L_2,F\right]&=L_3L_2-L_2L_1+i_2,\\
\left[L_3,F\right]&=L_1L_3-L_3L_2+i_3,
\end{split}
\end{align}
where $i_1, i_2, i_3$ are structure constants (or central elements). The first work on (\ref{eq:eqRacah}) can be found in \cite{Gao&Wang&Hou-2013}. The details of the passage between \eqref{eq:Racah} and \eqref{eq:eqRacah} can also be found in \cite{Genest&Vinet&Zhedanov-2014}.

The Racah algebra arises in the context of superintegrable models. A quantum system described by a Hamiltonian $H$ with $d$ degrees of freedom is said to be maximally superintegrable if it admits $2d-1$ algebraically independent constants of motion, including $H$ itself. It is said to be superintegrable of second order if the conserved charges are of degree at most $2$ in the momenta. For a recent review, see \cite{Miller&Post&Winternitz-2013}. The superintegrable models of second order (without reflection operators) in two dimensions which admit separation of variables have been classified. They can all be obtained as limits or special cases of the so-called generic $3$-parameter system on the $2$-sphere \cite{Kalnins&Miller&Post-2013}. This model is governed by the Hamiltonian:
\begin{equation} \label{eq:Hamiltonian1}
 H:=\sum_{1\leq i < j \leq 3} J_{ij}^2
 +(x_1^2 + x_2^2 + x_{3}^2)\sum_{i=1}^{3}\frac{a_i}{x_i^2},
\end{equation}
where $J_{ij}=i(x_j \partial_{x_i} - x_i \partial_{x_j})$ are the angular momentum operators and where $a_1,a_2,a_3$ are real parameters. Of course, $x_1^2 + x_2^2 + x_{3}^2 = 1$ on the 2-sphere. The superintegrability of $H$ and the separation of variables of the associated Schr\"odinger equation is demonstrated in \cite{Harnad&Yermolayeva-2004} using the $R$-matrix formalism. The relevance of the Racah algebra in this context appeared when the symmetry algebra formed by the constants of motion of $H$ was realized in \cite{Kalnins&Miller&Post-2007} with the help of the difference operator of which the Wilson/Racah polynomials are eigenfunctions. A cogent analysis was subsequently provided in \cite{Genest&Vinet&Zhedanov-2014-2} in the framework of the Racah problem for $\mathfrak{su}(1,1)$. Upon using the singular oscillator realization of $\mathfrak{su}(1,1)$, it was observed  that the Hamiltonian $H$ can be identified with the total Casimir operator obtained when three such representations are combined. The symmetry algebra was then identified noting that it is generated by the intermediate Casimir operators which satisfy the relations of the (centrally extended) Racah algebra. This also provided an explanation for the occurrence of the Racah polynomials as overlap coefficients between wavefunctions separated in different spherical coordinate systems. The reader is referred to \cite{Genest&Vinet&Zhedanov-2014-3} for a review of this picture.  

It has been shown that the (equitable presentation of the) Racah algebra can also be obtained from quadratic elements in the enveloping algebra of $\mathfrak{su}(2)$ \cite{Gao&Wang&Hou-2013}, and the correspondence with the recoupling scheme has been given in \cite{Genest&Vinet&Zhedanov-2014}. Progress towards obtaining the symmetry algebra of Hamiltonian analogous to \eqref{eq:Hamiltonian1} on spheres of arbitrary dimensions and identifying the structure of the Racah algebra for higher rank has been made in \cite{Kalnins&Miller&Post-2011}, where the model on the $3$-sphere was examined and bases for the symmetry algebra representations obtained in terms of Tratnik's bivariate Racah polynomials \cite{Tratnik-1991}. Furthermore, an extension of the approach introduced in \cite{Genest&Vinet&Zhedanov-2014-2} and \cite{Genest&Vinet&Zhedanov-2014} was given in \cite{Post} to establish the connection between the recoupling of four copies of $\mathfrak{su}(1,1)$, the generic system on the 3-sphere, and these bivariate Racah polynomials. As this manuscript was being completed, an interesting characterization of the symmetry algebra of the generic superintegable system on the $n$-sphere in relation to the Kohno-Drinfeld algebra was produced \cite{Iliev-2016}.

Germane to the present analysis are the recent developments on the interconnectedness of Bannai--Ito algebras, recoupling of multifold tensor product representations of $\mathfrak{osp}(1,2)$, superintegrable systems with reflections, and the Dunkl operators associated to the $\mathbb{Z}_2^{n}$ root system. The (rank one) Bannai--Ito algebra was introduced in \cite{Tsujimoto&Vinet&Zhedanov-2012} to describe the bispectrality of the univariate Bannai--Ito polynomials. This associative algebra over $\mathbb{C}$ has three generators $A_i$ with $i=1,2,3$ satisfying
\begin{align}
\label{eq:Bannai-Ito} 
\{A_1,A_2\}=A_3+\omega_3 \quad \{A_2, A_3\}=A_1+\omega_1 \quad \{A_1, A_3\}=A_2+\omega_2,
\end{align}
where $\omega_1$, $\omega_2$, $\omega_3$ are structure constants; the algebra \eqref{eq:Bannai-Ito} can be viewed as the $q=-1$ case of the Askey--Wilson algebra. In \cite{Genest&Vinet&Zhedanov-2014-4}, the Racah coefficients of $\mathfrak{osp}(1,2)$ were found to be expressible in terms of the Bannai-Ito polynomials and the intermediate Casimir operators for three-fold tensor product representations of $\mathfrak{osp}(1,2)$ were correspondingly seen to verify a central extension of \eqref{eq:Bannai-Ito}. A superintegrable Hamiltonian involving reflection operators with the Bannai--Ito algebra as symmetry algebra was obtained in this framework by combining three parabosonic or Dunkl oscillator realizations of $\mathfrak{osp}(1,2)$ \cite{Genest&Vinet&Zhedanov-2014-5}. Related is the observation that the commutant of the Dunkl Laplacian on $\mathbb{S}^2$ also leads to a central extension of the Bannai-Ito algebra with the consequence that the space of Dunkl harmonics supports representations of the (rank one) Bannai--Ito algebra \cite{Genest&Vinet&Zhedanov-2015}. It is worth mentioning that embeddings of the Racah algebra into the Bannai--Ito algebra have been constructed \cite{Genest&Vinet&Zhedanov-2015-2}. The Dirac--Dunkl equation on the $2$-sphere has been analyzed in the same spirit as the scalar Laplace-Dunkl equation. It was also found to admit the Bannai-Ito algebra as its symmetry algebra \cite{DeBie&Genest&Vinet-2016}. This model proved instructive and somewhat simpler. It exhibited a lower computational complexity and brought on tools such as the Cauchy-Kovalevskaia map that proved quite useful to obtain separated wavefunctions as bases for representation spaces. The extension to arbitrary dimensions of the analysis of the symmetries of the Dirac-Dunkl equation on spheres has been carried out in \cite{DeBie&Genest&Vinet-2016-2} and it allowed to identify the structure of the higher rank Bannai-Ito algebra. It also paved the way towards the identification of the multivariable Bannai-Ito polynomials. 

As stated before, the objective of this paper is to similarly obtain the structure of the higher rank Racah algebra. The organization of the paper is as follows. In Section 2, the connection between the Laplace--Dunkl operator associated to $\mathbb{Z}_2^n$ and the Lie algebra $\mathfrak{su}(1,1)$ is reviewed. In Section 3, the rank one Racah algebra in the equitable presentation is displayed. In Section 4, the higher rank Racah algebra is constructed by examining multifold tensor product representations of $\mathfrak{su}(1,1)$. The connections with the Laplace-Dunkl operator for $\mathbb{Z}_2^{n}$, as well as with the superintegrable system on the $n$-sphere, are established. A Cauchy-Kovalevskaia theorem is used to construct bases for the $\mathbb{Z}_2^n$ Dunkl harmonics in Section 5. In Section 6, the connection coefficients between these bases are expressed in terms of Tratnik's multivariate Racah polynomials. A brief conclusion follows.

\section{$\mathbb{Z}_2^n$ Dunkl operators and $\mathfrak{su}(1,1)$}\label{DunklOperators}

Consider the Abelian reflection group $\mathbb{Z}_2^{n}=\mathbb{Z}_2\times \cdots \times \mathbb{Z}_2$. The corresponding Dunkl operators $T_1, \ldots, T_{n}$ acting on $\mathbb{R}^{n}$ are defined as follows:
\begin{align*}
T_i=\partial_{x_i}+\frac{\mu_i}{x_i}(1-r_i),\qquad i=1,\ldots,n,
\end{align*}
where $\mu_1, \ldots,\mu_n$ with $\mu_i> 0$ are real parameters, $\partial_{x_i}$ is the partial derivative with respect to $x_i$ and $r_i$ is the reflection operator in the $x_i=0$ hyperplane, i.e. $r_i f(x_i)=f(-x_i)$. It is obvious that one has $T_i T_j=T_j T_i$ for all $i,j\in \{1,\ldots, n\}$. The Laplace--Dunkl operator $\Delta$ associated to $\mathbb{Z}_2^{n}$ is defined by
\begin{align}\label{Laplace-Dunkl}
\Delta=\sum_{i=1}^{n}T_i^{2}.
\end{align}
The ``intermediate'' Laplace-Dunkl operators and squared position operators are defined as follows. Let $A\subset [n]$, where $[n]=\{1,\ldots ,n\}$, and define
\begin{align*}
\Delta_{A}=\sum_{i\in A}T_i^{2},\qquad \rVert x_{A}\rVert^{2}=\sum_{i\in A}x_i^{2}.
\end{align*}
In this notation, the ``full'' $n$-dimensional Laplace-Dunkl operator  $\Delta$ given in \eqref{Laplace-Dunkl} can be written as $\Delta_{[n]}$. For $\ell\leq n$, we shall also use the notation $\Delta_{[\ell]}$ and $\rVert x_{[\ell]}\rVert^{2}$ for $\Delta_{\{1,\ldots, \ell\}}$ and $\rVert x_{\{1,\ldots,\ell\}}\rVert^{2}$, respectively. 

The space $\mathcal{H}_k(\mathbb{R}^{n})$ of Dunkl harmonics of degree $k$ is defined as $\mathcal{H}_k(\mathbb{R}^{n}) :=\mathrm{Ker}\,\Delta\, \cap\, \mathcal{P}_k(\mathbb{R}^{n})$, where $\mathcal{P}_k(\mathbb{R}^{n})$ stands for the space of homogeneous polynomials of degree $k$ in $\mathbb{R}^{n}$. The following decomposition of homogeneous polynomials in terms of Dunkl harmonics is called the Fischer decomposition.
\begin{proposition}[\cite{Dunkl-2001}]
	\label{fischerdecomp}
	The space $\mathcal{P}_k(\mathbb{R}^{n})$ of homogeneous polynomials of degree $k$ has the direct sum decomposition
	\begin{align*}
	\mathcal{P}_k(\mathbb{R}^{n})=\bigoplus_{j=0}^{\lfloor \frac{k}{2}\rfloor}\| x_{[n]} \|^{2j}\,\mathcal{H}_{k-2j}(\mathbb{R}^{n}).
	\end{align*}
\end{proposition}

An important observation for the purposes of this paper, due to Heckman \cite{Heckman-1991}, is that the Dunkl-Laplace operator, together with the squared position operator and the dilation operator, realize the $\mathfrak{su}(1,1)$ relations.
\begin{proposition}[\cite{Heckman-1991}]
\label{su2prop}
Let $A \subset [n]$ and define
\begin{align}
\label{SU11}
A_0 	= 	\frac{1}{2}(\mathbb{E}_{A}+\gamma_{A}),\qquad
J_{+} 	=	\frac{\rVert x_{A}\rVert^{2}}{2},\qquad
J_{-}	=	\frac{\Delta_{A}}{2},
\end{align}
where $\mathbb{E}_{A}=\sum_{i\in A}x_i\partial_{x_i}$ is the Euler (or dilation) operator for the set $A$, and where 
\begin{align}
\label{gamma-A}
\gamma_A=\frac{\rvert A\rvert}{2}+\sum_{i\in A} \mu_i.
\end{align}
The operators $A_0$, $J_{\pm}$ satisfy the $\mathfrak{su}(1,1)$ relations
\begin{align}
[A_0, J_{\pm}] = \pm J_{\pm},\qquad [J_{-}, J_{+}] = 2 A_0.
\end{align}
\end{proposition}
The Casimir operator for $\mathfrak{su}(1,1)$, which commutes with $A_0$ and $J_{\pm}$, is given by
\begin{align}
\label{su-Casimir}
C := A_0^2 - J_{+}J_{-} - A_0.
\end{align}
The Casimir operator $C_A$ associated to the $\mathfrak{su}(1,1)$ realization \eqref{SU11} has the expression
\begin{align}
\label{Casimir}
C_A = \frac{1}{4}\left( \left( \mathbb{E}_A + \gamma_A\right)^2-2\left( \mathbb{E}_A + \gamma_A\right)- \| x_A \|^2 \Delta_A\right).
\end{align}
By construction, one has
\begin{equation}
\left[ C_A, \Delta_A \right]=0, \quad \left[ C_A, \rVert x_{A}\rVert^{2} \right]=0.
\end{equation}
The following lemma is easily derived from the above observations.
\begin{lemma}
For $A \subset [n]$, the operator $C_A$ satisfies
\begin{equation*}
\left[C_A, \Delta_{[n]} \right]=0,
\end{equation*}
\end{lemma}
Hence for every $A \subset [n]$, $C_{A}$ is a symmetry of $\Delta_{[n]}$, the Laplace-Dunkl operator associated to the root system $\mathbb{Z}_2^{n}$. It is clear that for two generic sets $A\neq B$, the operators $C_{A}$ and $C_{B}$ will not commute. The symmetry algebra they generate will be obtained below using the connection between $\Delta_{[n]}$ and the recoupling of $\mathfrak{su}(1,1)$ representations. This invariance algebra will also be shown to correspond to that of the generic system on the $n$-sphere. Before tackling the higher rank case, let us quickly review the presentation of the Racah algebra in the rank one case.

\section{The rank one Racah algebra}
The universal (centrally extended) Racah algebra of rank $1$ has $7$ generators. The generators $C_a$, $C_b$, $C_c$, and $C_{abc}$ are central. The three other generators are denoted $C_{ab}$, $C_{bc}$ and $C_{ac}$. The generators are related to one another as follows
\begin{align}\label{A}
C_{abc}=C_{ab}+C_{bc}+C_{ac}-C_a-C_b-C_c.
\end{align}
To write down the relations of the Racah algebra, we introduce an extra element $F$ defined as
\begin{align}\label{B}
2F:=[C_{ab},C_{bc}] = [C_{ac},C_{ab}] = [C_{bc},C_{ac}].
\end{align}
In light of \eqref{A}, the definition \eqref{B} is consistent. The non-trivial commutation relations have the expression 
\begin{align} \label{RankoneRacah}
\begin{split} \\
[C_{ab},F]&=C_{bc}C_{ab}-C_{ab}C_{ac}+\left(C_b-C_a\right)\left(C_{c}-C_{abc}\right), \\
[C_{bc},F]&=C_{ac}C_{bc}-C_{bc}C_{ab}+\left(C_c-C_b\right)\left(C_{a}-C_{abc}\right),\\
[C_{ac},F]&=C_{ab}C_{ac}-C_{ac}C_{bc}+\left(C_a-C_c\right)\left(C_{b}-C_{abc}\right).
\end{split}
\end{align}
This algebra agrees with the equitable presentation \eqref{eq:eqRacah} where the $L_i$ play the role of the $C_{xy}$ generators. The structure constants $ i_1$, $i_2$ and $i_3$ are here replaced by central elements.

\section{Tensor algebra approach to a higher rank Racah algebra}
In \cite{Genest&Vinet&Zhedanov-2014-2, Genest&Vinet&Zhedanov-2015-2}, the rank-one Racah algebra was seen to arise in considering three-fold tensor product representations of $\mathfrak{su}(1,1)$. We shall here generalize this approach by taking the $n$-fold product. In \cite{Lehrer&Zhang-2006} $n$-fold tensor products of Lie algebras have been studied in a more general setting. Consider the $\mathfrak{su}(1,1)$ comultiplication $\Delta$
\begin{equation}
\label{comultiplication} 
\Delta(A_0)=A_0\otimes 1+1\otimes A_0, \qquad \Delta(J_\pm)=J_\pm\otimes 1+1\otimes J_\pm .
\end{equation}
The comultiplication is an algebra morphism which extends to the universal enveloping algebra. Hence, it is possible to calculate the comultiplication of any element of $\mathcal{U}(\mathfrak{su}(1,1))$.
Define the following elements:
$$
\mathbf{C}_1:=C, \qquad  \mathbf{C}_n:=(\underbrace{1\otimes\ldots\otimes 1}_{n-2 \text{ times }} \otimes \Delta)(\mathbf{C}_{n-1}),
$$
where $C$ is the Casimir operator of $\mathfrak{su}(1,1)$ given in \eqref{su-Casimir}. Consider the map
$$
\tau_k: \bigotimes_{i=1}^{m-1} \mathcal{U}(\mathfrak{su}(1,1))\rightarrow \bigotimes_{i=1}^{m} \mathcal{U}(\mathfrak{su}(1,1)),
$$
which acts as follows on the homogeneous tensor products:
$$
\tau_k(t_1\otimes \ldots \otimes t_{m-1}):=t_1\otimes \ldots \otimes t_{k-1} \otimes 1 \otimes t_k \otimes \ldots \otimes t_{m-1} .
$$
and extend it by linearity.  The map $\tau$ adds a $1$ at the $k$-th place. This allows to define the following:
\begin{align}
\label{Casimir-Upper}
C^A:=\left(\prod_{k \in \left[ n\right] \backslash A}^{\longrightarrow} \tau_k \right)\left(\mathbf{C}_{|A|}\right).
\end{align}
with $A$ a subset of $[n]$. The set $A$ is the superindex whereas in the expression (\ref{Casimir}) it is the subindex. Notice that the order in which the $\tau_k$ act, matters. Take for example $n=5$ and $A=\left\{1,3\right\}$. To have a $1$ on the fourth place in the tensor product one needs to act with $\tau_4$. If one then acts by $\tau_k$, $k<4$, that $1$ will move to the next place. To avoid this, the proper order in our example must be 
\[
\prod_{k \in \left[ n\right] \backslash A}^{\longrightarrow} \tau_k =\tau_5\tau_4\tau_2.
\] 
One could ask if the order of the indices in $A$ matters.  By looking at the explicit form of, for instance, $\mathbf{C}_2$:
$$
\mathbf{C}_2=\Delta(C)=C\otimes 1+1\otimes C+2 A_0\otimes A_0-J_+\otimes J_--J_-\otimes J_+,
$$
one can see that switching the first and second positions does not change $\mathbf{C}_2$. This is true in general owing to the cocommutativity of the coproduct $\Delta$.

We define the generalized Racah algebra, denoted by $R(n)$, as the algebra generated by the $C^A$ with $A\subset[n]$. This set of generators is not independent, as can be seen from the relation
\begin{align}
\label{Relation}
C^A=\sum_{\left\{i,j\right\}\subset A} C^{ij}-\left(|A|-2\right)\sum_{i \in A} C^i .
\end{align}
In light of \eqref{Relation}, it suffices to exhibit the commutation relations between the $C^{A}$ with $|A|=2$ to fully describe the Racah algebra in terms of generators and relations. For convenience, we shall write $C^{ij}$ instead of $C^{\{i,j\}}$ for these elements. Similarly, we shall write $C^i$ instead of $C^{\{i\}}$ for the elements with only one index; it is easily seen that these elements are central. One observes that two elements with no index in common commute, e.g. $C^{12}C^{34}=C^{34}C^{12}$. Assume $i$, $j$ and $k$ to be distinct indices. Let $F^{ijk}$ be defined as
$$
F^{ijk}:=\frac{1}{2}\left[ C^{ij},C^{jk} \right].
$$
For $n=3$, the following relation completes the relation set that defines the Racah algebra: 
$$
\left[ C^{jk},F^{ijk}\right]=C^{ik}C^{jk}-C^{jk}C^{ij}+\left(C^k-C^j\right)\left(C^i-C^{ijk}\right).
$$
One notes that the above relations coincide in three dimensions with those of the centrally extended rank one Racah algebra that are obtained by taking $A$ to be the set $\{a,b,c\}$ in \eqref{RankoneRacah}. In higher dimension there are additional commutators to calculate. If $n\geq 4$, one finds the extra relation
$$\left[ F^{ijk},C^{kl}\right]=\left(C^{il}-C^i-C^l\right)\left(C^{jk}-C^j-C^k\right)-\left(C^{ik}-C^i-C^k\right)\left(C^{jl}-C^j-C^l\right).
$$
Let us simplify the above formula by introducing $P^{ij}:=C^{ij}-C^i-C^j$. In terms of the operators $P^{ij}$, the expression for the Casimir operators $C^{A}$ becomes
$$
C^A=\sum_{\left\{i,j\right\}\subset A} P^{ij}+\sum_{i \in A} C^i,
$$
and the commutation relations take the form
\begin{align*} 
&\left[ P^{ij},P^{jk} \right] = 2 F^{ijk}, \\
&\left[ P^{jk},F^{ijk}\right]= P^{ik}P^{jk}-P^{jk}P^{ij}+2P^{ik}C^j-2P^{ij}C^k,\\
&\left[ P^{kl},F^{ijk}\right]=P^{ik}P^{jl}-P^{il}P^{jk}.
\end{align*}
When $n\geq 5$, there are two additional commutators to calculate. One finds
\begin{align*}
&\left[ F^{ijk},F^{jkl} \right]=F^{jkl}P^{ij}-F^{ikl}\left(P^{jk}+2C^j\right)-F^{ijk}P^{jl}, \\
&\left[ F^{ijk},F^{klm} \right]=F^{ilm}P^{jk}-P^{ik}F^{jlm} .
\end{align*}
These relations are sufficient to fully describe the generalized Racah algebra in any dimension. These results are synthesized in the following proposition.
\begin{proposition}
The generalized Racah algebra $R(n)$ is generated by the Casimir operators $C^A$ defined in \eqref{Casimir-Upper}, where $A\subset [n]$. In light of the relation \eqref{Relation} between the generators, the defining relations of $R(n)$ can be presented only in terms of $C^{ij}$ with $i\neq j$, and the central elements $C^i$. Let $P^{ij}$ and $F^{ijk}$ be defined by
\begin{align*}
P^{ij} = C^{ij}-C^i-C^j,\qquad F^{ijk} = \frac{1}{2} \left[ P^{ij},P^{jk} \right].
\end{align*}
The defining relations of the Racah algebra $R(n)$ are
\begin{align}
\label{Gen-Racah}
\begin{aligned} 
&\left[ P^{ij},P^{jk} \right] = 2 F^{ijk}, \\
&\left[ P^{jk},F^{ijk}\right]= P^{ik}P^{jk}-P^{jk}P^{ij}+2P^{ik}C^j-2P^{ij}C^k,\\
&\left[ P^{kl},F^{ijk}\right]=P^{ik}P^{jl}-P^{il}P^{jk}.  \\
&\left[ F^{ijk},F^{jkl} \right]=F^{jkl}P^{ij}-F^{ikl}\left(P^{jk}+2C^j\right)-F^{ijk}P^{jl}, \\
&\left[ F^{ijk},F^{klm} \right]=F^{ilm}P^{jk}-P^{ik}F^{jlm},
\end{aligned}
\end{align}
where $i$, $j$, $k$, $l$ and $m$ are pairwise distinct indices in the set $[n]$. 
\end{proposition}
\begin{proof}
	The result follows by direct computations.
\end{proof}
The next lemma follows from the basic definitions
\begin{lemma}\label{commutators}
$\left[C^A,C^B\right]=0$ if $A \subset B$, $B\subset A$ or $A\cap B=\emptyset$.
\end{lemma}
Consider the tower of sets $[2]\subset[3]\subset \dots \subset [n-1]$. Because of lemma \ref{commutators} the generators corresponding to this tower of sets are mutually commutative. This yields an Abelian subalgebra denoted by
\begin{align}\label{LabellingAbelianAlgebra}
 \widehat{\mathcal{Y}}_n=\langle C^{[2]}, C^{[3]},\ldots, C^{[n-1]}\rangle 
\end{align}
which allows to label representation vectors as shall be seen. To every tower of sets one can associate such an Abelian subalgebra. It can be explicitly (and independently) verified that the operators $C_A$ of formula (\ref{Casimir}) in section \ref{DunklOperators} realize the generalized Racah algebra $R(n)$. To distinguish the abstract algebra from the Dunkl realization, we use upper and lower indices respectively on the symbols denoting the generators and their representations.
\begin{proposition}
The differential-difference operators $C_A$, given by \eqref{Casimir}, satisfy
\begin{align}
\label{rel}
C_A=\sum_{\left\{i,j\right\}\subset A} C_{ij}-\left(|A|-2\right)\sum_{i \in A} C_i .
\end{align}
Putting $P_{ij} = C_{ij} -C_i - C_j$, and $F_{ijk} = \frac{1}{2} \left[ P_{ij},P_{jk} \right]$, it is verified that these operators realize the generalized Racah algebra \eqref{Gen-Racah}.
\end{proposition}

\begin{proof}
From formula (\ref{Casimir}) we have
\begin{align}
\label{Exp-Cas-1}
C_i = \frac{1}{4} \left(\mu_i^2-\mu_i r_i - \frac{3}{4}\right),\qquad 
C_{ij} = \frac{1}{4} \left( - L_{ij}^2 + (\mu_i r_i + \mu_j r_j)^2 -1 \right),
\end{align}
with $L_{ij} = x_i T_j -x_j T_i$ the Dunkl angular momentum operators. Using 
\begin{align*}
[L_{ij}, L_{jk}]= L_{ik}(1+ 2 \mu_j r_j),
\end{align*}
one finds
\begin{align*}
[C_{ij}, C_{jk}]& = \frac{1}{16} [L_{ij}^2, L_{jk}^2]\\
& = \frac{1}{8} \left(L_{ij}^2 (1+ 2 \mu_k r_k) -  L_{ik}^2 (1+ 2 \mu_j r_j) -  L_{jk}^2 (1+ 2 \mu_i r_i) + 2 L_{ik} L_{ij} L_{jk} \right)\\
& := 2 F_{ijk}.
\end{align*}
The other relations follow after long but straightforward computations, again using $P_{ij} = C_{ij} -C_i - C_j$ to formulate the end result.
\end{proof}
Note that operators of the type $L^2_{ij}$ have also appeared in \cite{Dunkl-1999}, in the context of orthogonal bases.

\subsection{Connection with the generic system on the $n$-sphere}
The generalized Racah algebra $R(n)$, defined by the relations \eqref{Gen-Racah}, is also the symmetry algebra of the generic superintegrable system on the $n$-sphere. This can be shown as follows. Consider the Casimir operator $C_{[n]}$, as given in \eqref{Casimir}, and define the operator $\widetilde{C}_{[n]}$ via the gauge transformation
\begin{align*}
\widetilde{C}_{[n]} := \mathcal{G}(x_1,\ldots, x_{n}) C_{[n]} \mathcal{G}^{-1}(x_1,\ldots,x_{n}),
\end{align*}
where the gauge factor is given by $\mathcal{G}(x_1,\ldots, x_{n}) = |x_1|^{\mu_1}\cdots |x_{n}|^{\mu_n}$. It is directly verified using \eqref{Exp-Cas-1} that $\widetilde{C}_{[n]}$ has the expression
\begin{align}
\label{Ham-2}
\widetilde{C}_{[n]} =\frac{1}{4}\left( \sum_{1\leq i<j\leq n}J_{ij}^2 + (x_1^2 + \cdots + x_{n}^2)\sum_{k=1}^{n}\frac{\mu_k(\mu_k-r_k)}{x_k^2} + \frac{n(n-4)}{4}\right).
\end{align}
Upon comparing the above expression with \eqref{eq:Hamiltonian1}, it is seen that \eqref{Ham-2} corresponds to the generic system on the $n$-sphere, with the particularity of having reflection operators appear in the ``parameters'' $a_i = \mu_i(\mu_i-r_i)$ for $i=1,\ldots,n$. However, since $\widetilde{C}_{[n]}$ commutes with all reflections $r_1,\ldots, r_n$, one can diagonalize these reflections simultaneously with $\widetilde{C}_{[n]}$, so that for each of the $2^{n}$ parity sectors, one gets a bona fide (scalar) generic system on the $n$-sphere. For example, in the totally even sector, one has $a_i = \mu_i(\mu_i -1)$ for $i=1,\ldots,n$. 

In light of Lemma 2, the Hamiltonian defined by \eqref{Ham-2} is clearly superintegrable, and its symmetries are given by the gauge-transformed Casimir operators $\widetilde{C}_{A}$ with $A\in [n]$. The relation \eqref{rel} holds, and one has $\widetilde{C}_i = C_i$ for the symmetries with one index. The symmetries $\widetilde{C}_{ij}$ with two indices are given by
\begin{align*}
\widetilde{C}_{ij} = \frac{1}{4}\left( J_{ij}^2 + (x_i^2 + x_j^2) \left\{\frac{\mu_i(\mu_i-r_i)}{x_i^2}+ \frac{\mu_j(\mu_j-r_j)}{x_j^2}\right\}-1\right),
\end{align*}
and they satisfy the commutation relations \eqref{Gen-Racah} of the generalized Racah algebra.
\begin{remark}
As observed in \cite{Iliev-2016}, the intermediate Casimir operators $\widetilde{C}_{ij}$ realize the Drinfeld-Kohno relations, that is
\begin{align*}
[\widetilde{C}_{ij}, \widetilde{C}_{kl}] = 0,\qquad [\widetilde{C}_{ij}, \widetilde{C}_{ik} + \widetilde{C}_{jk}]=0,
\end{align*}
when $i$, $j$, $k$, $l$ and $m$ are pairwise distinct indices in the set $[n]$.
\end{remark}
\subsection{Embeddings of $R(3)$ into $R(n)$}
Let us now exhibit how the rank-one Racah algebra $R(3)$ can be embedded in the generalized Racah algebra $R(n)$. For a given set of indices $K$, we have already showed how to lift the Casimir of $\mathfrak{su}(1,1)$ to an intermediate Casimir operator $C^K$ acting on the $n$-fold tensor product. In similar fashion, one can lift the generators $A_0$ and $J_\pm$ of $\mathfrak{su}(1,1)$ to the $n$-fold tensor product. Let $A_{0,n}$ and $J_{\pm,n}$ be defined as
\begin{align*} 
A_{0,n}:=(\underbrace{1\otimes\ldots\otimes 1}_{n-2 \text{ times }} \otimes \Delta)(A_{0,n-1}) \\
J_{\pm,n}:=(\underbrace{1\otimes\ldots\otimes 1}_{n-2 \text{ times }} \otimes \Delta)(J_{\pm,n-1}),
\end{align*}
and define $A_0^{K}$, $J_{\pm}^{K}$ by applying the $\tau_k$ map
\begin{align*}
A_0^K:=\left(\prod_{k \in \left[ n\right] \backslash K}^{\longrightarrow} \tau_k \right)\left(A_{0,|K|}\right),\qquad
J_\pm^K:=\left(\prod_{k \in \left[ n\right] \backslash K}^{\longrightarrow} \tau_k \right)\left(J_{\pm,|K|}\right).
\end{align*}
The elements $A_0^K$ and $J_{\pm}^K$ generate $\mathfrak{su}^K(1,1)$, a $|K|$-fold tensor product representation of $\mathfrak{su}(1,1)$ acting on the components $k_i\in K\subset [n]$; the associated Casimir operator is $C^{K}$. Consider three pairwise disjoint subsets of $[n]$ and call them $K$, $L$ and $M$. One has the isomorphism
$$
\langle \mathfrak{su}^K(1,1),\mathfrak{su}^L(1,1),\mathfrak{su}^M(1,1)\rangle \cong \mathfrak{su}(1,1)\otimes\mathfrak{su}(1,1)\otimes\mathfrak{su}(1,1).
$$
The generators $X=A_0, J_{\pm}$ are mapped in the following way
\begin{align*} 
X^K \rightarrow X \otimes 1\otimes 1,\quad
X^L \rightarrow 1\otimes X \otimes 1, \quad
X^M \rightarrow 1\otimes 1\otimes X.
\end{align*}
This gives embeddings of the threefold tensor product into the $n$-fold tensor product. Similarly, for each triple of pairwise disjoint subsets of $[n]$, called $K$, $L$, and $M$, one has a realization of the Racah algebra $R(3)$. In this realization, the central elements are $C^K$, $C^L$, $C^M$, as well as $C^{K\cup  L \cup M}$, and the non-commuting operators are $C^{K\cup L}$, $C^{K \cup M}$ and $C^{L\cup M}$. For example, the analog of relation \eqref{A} is 
\begin{align*}
C^{K\cup  L \cup M} = C^{K\cup L} + C^{K \cup M} + C^{L\cup M} - C^K - C^{L} - C^{M}.
\end{align*}
The subalgebras of $R(n)$ generated in this way shall be denoted by $R^{K,L,M}(3)$

\section{A basis for the space of Dunkl harmonics}
We shall now construct a basis for the space of Dunkl harmonics using the Cauchy-Kovalevskaia isomorphism.

\subsection{The CK isomorphism} 
For $i \in [n]$, let $Q_i^{\pm}$ denote the projection operators
$$
Q_i^{\pm}= \frac{1}{2}\left( 1 \pm r_i \right),
$$
which satisfy
$$
\left(Q_i^{\pm} \right)^2 =Q_i^{\pm},\qquad  Q_i^{\pm} Q_i^{\mp} = 0.
$$
Upon taking defining $\mathcal{H}_k^{\pm}(\mathbb{R}^n) = Q_n^{\pm} \mathcal{H}_k(\mathbb{R}^n)$, one has the following decomposition  of the space of Dunkl harmonics
$$
\mathcal{H}_k(\mathbb{R}^n) = \mathcal{H}_k^+(\mathbb{R}^n) \oplus \mathcal{H}_k^-(\mathbb{R}^n).
$$
Note that one has also $r_n \mathcal{H}_k^{\pm}(\mathbb{R}^n) = \pm \mathcal{H}_k^{\pm}(\mathbb{R}^n)$. Let $h\in \mathcal{H}_k(\mathbb{R}^n)$, which means
\begin{align}
\label{dompe}
\Delta_{\left[ n \right]} h=0.
\end{align}
In proposition \ref{ActionofAbelianAlgebra}, the elements $h\in \mathcal{H}_k(\mathbb{R}^n)$ will also be seen to be eigenfunctions of the Casimir operator $C_{[n]}$; they thus correspond to wavefunctions of the system defined by \eqref{Ham-2}, up to a gauge transformation. Upon separating the variable $x_n$, one can write
\begin{align}
\label{first}
\Delta_{\left[ n \right]} =\Delta_{\left[ n-1 \right]} + T_n^2,
\end{align}
as well as
\begin{align}
\label{second}
h=\sum_{i=0}^k x_n^i p_{k-i},
\end{align}
with $p_{k-i} \in \mathcal{P}_{k-i}(\mathbb{R}^{n-1})$, a homogeneous polynomial of degree $k-i$ in the variables $x_1,\ldots, x_{n-1}$. Upon applying \eqref{first} on \eqref{second} given \eqref{dompe}, one finds
\begin{equation*}
 0=\sum_{j=2}^{k} \left[j\right]_n \left[j-1\right]_{n}x_n^{j-2}p_{k-j}+\sum_{j=0}^{k-2}x_n^j\Delta_{\left[ n-1 \right]} p_{k-j},
\end{equation*}
where
$$
\left[j\right]_n:= j+\mu_n(1-(-1)^j).
$$ 
Comparing the powers of $x_n$, it follows that
\[
\Delta_{\left[ n-1 \right]} p_{k-j} =-\left[j+2\right]_n \left[j+1\right]_{n}p_{k-j-2}.
 \]
This relation allows to obtain $p_{k-i}$ recursively. Given $p_k$ and $p_{k-1}$, one has
\begin{align*}
p_{k-2j}=\frac{(-1)^j\Delta^j_{[n-1]}}{2^{2j}j!\left(\mu_n+\frac{1}{2}\right)_{j}}p_k,\qquad 
p_{k-2j-1}=\frac{(-1)^j\Delta^j_{[n-1]}}{2^{2j}j!\left(\mu_n+\frac{3}{2}\right)_{j}}p_{k-1},
\end{align*}
where $ (x)_{n}=x(x+1)\ldots(x+n-1)$ is the raising Pochhammer symbol. With these expressions for $p_{k-i}$, one can write $h$ explicitly. Indeed, introduce the two maps
\begin{align}
\mathbf{CK}^0_{x_n}:=\sum_{j=0}^{\lfloor \frac{k}{2} \rfloor}  \frac{(-1)^jx_n^{2j}\Delta_{\left[n-1\right]}^j}{2^{2j}j!\left(\mu_n+\frac{1}{2}\right)_{j}} ,\qquad 
\mathbf{CK}^1_{x_n}:=\sum_{j=0}^{\lfloor \frac{k}{2} \rfloor}  \frac{(-1)^jx_n^{2j+1}\Delta_{\left[n-1\right]}^j}{2^{2j}j!\left(\mu_n+\frac{3}{2}\right)_{j}}.
\end{align}
It is clear that $h=\mathbf{CK}^0_{x_n}(p_k)+\mathbf{CK}^1_{x_n}(p_{k-1})$. This leads to the following proposition.
\begin{proposition}
\label{CKmaps}
There exist isomorphisms $\mathbf{CK}^0_{x_n}: \mathcal{P}_{k}(\mathbb{R}^{n-1}) \rightarrow\mathcal{H}_k^+(\mathbb{R}^n)$ and  $\mathbf{CK}^1_{x_n}: \mathcal{P}_{k-1}(\mathbb{R}^{n-1}) \rightarrow\mathcal{H}_k^-(\mathbb{R}^n)$ which are given by
\begin{align}
\label{CKplus}
\mathbf{CK}^0_{x_n} & = \Gamma(\mu_n+1/2)\widetilde{I}_{\mu_n-1/2}(x_{n}\sqrt{\Delta_{\left[n-1\right]}}),\\
\label{CKminus}
\mathbf{CK}^1_{x_n} & = \Gamma(\mu_n+3/2) x_n \widetilde{I}_{\mu_n+1/2}(x_n\sqrt{\Delta_{\left[n-1\right]}}),
\end{align}
with $\widetilde{I}_{\delta}(x)=(\frac{2}{x})^{\delta}I_{\delta}(x)$ and $I_{\delta}(x)$ the modified Bessel function \cite{Andrews&Askey&Roy-2001}. 
\end{proposition}

\begin{proof}
The explicit expressions for $\mathbf{CK}^0_{x_n}$ and $\mathbf{CK}^1_{x_n}$ are given above. By considering the hypergeometric expansion of the Bessel function one easily finds (\ref{CKplus}) and (\ref{CKminus}). $\mathbf{CK}^0_{x_n}$ will map to $\mathcal{H}_k(\mathbb{R}^d)$ by construction. Since all monomials in $\mathbf{CK}^0_{x_n}(p_k)$ are of even degree in $x_n$, it will map to $\mathcal{H}_k^+(\mathbb{R}^n)$.  One straightforwardly proves the equivalent for $\mathbf{CK}^{1}_{x_n}$. It remains to prove that these maps are isomorphisms. First we prove injectivity. Let $p_k \in \mathcal{P}_k^+(\mathbb{R}^{n-1})$ and $p_{k-1} \in \mathcal{P}_{k-1}^-(\mathbb{R}^{n-1})$. One has
\begin{align*}
\mathbf{CK}^0_{x_n}(p_k)\mid_{x_n=0}=p_k, \qquad
\frac{\partial \mathbf{CK}^1_{x_n}(p_{k-1})}{\partial x_n}\mid_{x_n=0}=p_{k-1}.
\end{align*}
Hence, the maps must be injective. For surjectivity we use an argument on the dimensions. Let $h$ be in $\mathcal{H}_k(\mathbb{R}^{n})$. We map this element onto the element 
$$
 (p_k,p_{k-1}):=(h\mid_{x_n=0},\frac{\partial h}{\partial x_n}\mid_{x_n=0}) \in \mathcal{P}_k(\mathbb{R}^{n-1})\oplus\mathcal{P}_{k-1}(\mathbb{R}^{n-1}).
$$
By the construction above, we know that $h=\mathbf{CK}^0_{x_n}(p_k)+\mathbf{CK}^1_{x_n}(p_{k-1})$. Hence, this map must be injective. 
This means that 
$$
\dim(\mathcal{H}_k(\mathbb{R}^{n}))\leq \dim(\mathcal{P}_k(\mathbb{R}^{n-1})\oplus\mathcal{P}_{k-1}(\mathbb{R}^{n-1})).
$$
The injectivity of the $\mathbf{CK}$ maps gives the following inequality.
$$
 \dim(\mathcal{P}_k(\mathbb{R}^{n-1})\oplus\mathcal{P}_{k-1}(\mathbb{R}^{n-1})) \leq \dim(\mathcal{H}^+_k(\mathbb{R}^{n}))+\dim(\mathcal{H}^-_k(\mathbb{R}^{n}))=\dim(\mathcal{H}_k(\mathbb{R}^{n})),
$$
Combining these two inequalities, we see that they must be equalities. The dimensions of $\mathcal{P}_k(\mathbb{R}^{n-1})$ and $\mathcal{P}_{k-1}(\mathbb{R}^{n-1})$ are equal to the dimensions of $\mathcal{H}^+_k(\mathbb{R}^{n})$ and $\mathcal{H}^-_k(\mathbb{R}^{n})$ respectively. The $\mathbf{CK}$ maps are isomorphisms.
\end{proof}

\begin{remark}
In \cite{DeBie&Genest&Vinet-2016, DeBie&Genest&Vinet-2016-2}, a CK isomorphism between the space of Clifford-valued homogeneous polynomials in $\mathbb{R}^{n-1}$ and the space of polynomial homogeneous null-solutions of the Dirac-Dunkl operator in $\mathbb{R}^n$ was established.
\end{remark}

\subsection{Construction of the basis}

The $\mathbf{CK}$ maps can be combined with the Fischer decomposition (Proposition \ref{fischerdecomp}) to construct an explicit basis for the space $\mathcal{H}_{k}(\mathbb{R}^{n})$. Consider the following tower of $\mathbf{CK}$ extensions and Fischer decompositions:
\begin{align}
\mathcal{H}_{k}^+(\mathbb{R}^{n})	&	\cong\mathbf{CK}_{x_n}^{0}\big[\mathcal{P}_{k}(\mathbb{R}^{n-1})\big] \nonumber
\\\nonumber
	&\cong\mathbf{CK}_{x_n}^{0}\big[\bigoplus_{l=0}^{\lfloor\frac{k}{2}\rfloor}\|x_{[n-1]}\|^{2l}\mathcal{H}_{k-2l}(\mathbb{R}^{n-1})\big]
\\\nonumber
	&\cong\mathbf{CK}_{x_n}^{0}\Big[\bigoplus_{l=0}^{\lfloor\frac{k}{2}\rfloor}\|x_{[n-1]}\|^{2l}\big(\mathcal{H}_{k-2l}^+(\mathbb{R}^{n-1})\oplus\mathcal{H}_{k-2l}^-(\mathbb{R}^{n-1})\big)\Big]
\\\nonumber
	&\cong\mathbf{CK}_{x_n}^{0}\Big[\bigoplus_{l=0}^{\lfloor\frac{k}{2}\rfloor}\|x_{[n-1]}\|^{2l} \Big( \mathbf{CK}_{x_{n-1}}^{0}\big[ \mathcal{P}_{k-2l}(\mathbb{R}^{n-2})\big]\oplus\mathbf{CK}_{x_{n-1}}^{1}\big[ \mathcal{P}_{k-2l-1}(\mathbb{R}^{n-2})\big]\Big)\Big]
\\
\label{Tower}
	&\cong\mathbf{CK}_{x_n}^{0}\Big[\bigoplus_{l=0}^{\lfloor\frac{k}{2}\rfloor}\|x_{[n-1]}\|^{2l}\Big( \mathbf{CK}_{x_{n-1}}^{0}\big[ \bigoplus_{j=0}^{\lfloor\frac{k}{2}\rfloor-l}\|x_{[n-2]}\|^{2j}\mathcal{H}_{k-2l-2j}(\mathbb{R}^{n-2})\big]
\\\nonumber 
	&\hspace{40mm} \oplus\mathbf{CK}_{x_{n-1}}^{1}\big[ \bigoplus_{j=0}^{\lfloor\frac{k}{2}\rfloor-l}\|x_{[n-2]}\|^{2j}\mathcal{H}_{k-2l-2j-1}(\mathbb{R}^{n-2})\big]\Big)\Big]
\\\nonumber
	&\cong\cdots \nonumber
\end{align}
until one reaches the scalars. Since the scalars are spanned by $1$, one concludes the following.
\begin{proposition}
Let $\mathbf{l}$ be defined as $\mathbf{l}=(l_{n-1},l_{n-2},\ldots,l_{2}, l_{1})$ where $l_1,\ldots, l_{n}$ are non-negative integers. Let $\vec \epsilon = (\epsilon_n, \ldots, \epsilon_1)$ with $\epsilon_j \in \{ 0 , 1\}$. Consider the set of functions $\mathscr{Y}_{\mathbf{l}}^{\vec \epsilon}(x_1,\ldots, x_{n})$ defined by
\begin{multline}
\label{Basis-1}
\mathscr{Y}_{\mathbf{l}}^{\vec \epsilon}(x_1,\ldots, x_{n})=
\\
\mathbf{CK}_{x_{n}}^{\epsilon_n}\bigg[\|{x_{[n-1]}\|^{2l_{n-1}}}\mathbf{CK}_{x_{n-1}}^{\epsilon_{n-1}}\Big[\|x_{[n-2]}\|^{2l_{n-2}}\big[\cdots \mathbf{CK}_{x_2}^{\epsilon_{2}}[\|x_{[1]}\|^{2l_1}\mathbf{CK}_{x_1}^{\epsilon_{1}}[1]]\big]\Big]\bigg] .
\end{multline}
The functions $\mathscr{Y}_{\mathbf{l}}^{\vec \epsilon}$ for which $k= \sum_{i=1}^n\epsilon_i+2\sum_{i=1}^{n-1} l_i$ form a basis for the space $\mathcal{H}_{k}(\mathbb{R}^{n})$ of $k$-homogeneous Dunkl harmonics.
\end{proposition}

The upper index $\vec \epsilon$ also suggests the following direct sum decomposition of the harmonics of degree $k$.

\begin{proposition}
Let $\vec \beta = (\beta_n, \ldots, \beta_1)$ with $\beta_j \in \{ + , -\}$. One has
\[
\mathcal{H}_k(\mathbb{R}^n) = \bigoplus_{\vec \beta } \mathcal{H}_k^{\vec \beta}(\mathbb{R}^n),
\]
where the sum runs over all $\vec \beta \in \{ +,-\}^n$ and with
\[
\mathcal{H}_k^{\vec \beta}(\mathbb{R}^n):= \big\{ p \in \mathcal{H}_k(\mathbb{R}^n) \mid Q_i^{\beta_{i}}(p)=p, \quad i \in \{1, \dots, d\}\big\}.
\]

\end{proposition}

\begin{remark}
Note that for $k<n$ the spaces $H_k^{\vec \beta}(\mathbb{R}^n) $ can be empty for certain choices of $\vec \beta$. 
\end{remark}

\subsection{Explicit expression for the basis functions}
Let us now present a lemma which characterizes the action of the operators $\Delta$ and $\rVert x \rVert^2$ on the space of Dunkl harmonics.
\begin{lemma}
Let $h_{\ell}\in \mathcal{H}_{\ell}(\mathbb{R}^{n})$. For non-negative integers $j,k$ such that $j\leq k$ one has
\begin{align*}
 \Delta_{\left[n\right]}^j\rVert x_{\left[n\right]} \rVert^{2k}h_{\ell}&=\rVert x_{\left[n\right]} \rVert^{2(k-j)}h_\ell \prod_{t=k-j+1}^{k}  4t\left(l+t-1+\gamma_{\left[n\right]}\right) \\
 &=2^{2j}(-k)_j(-l-k+1-\gamma_{\left[n\right]} )_j \rVert x_{\left[n\right]} \rVert^{2(k-j)}h_\ell ,
\end{align*}
\end{lemma}

\begin{proof}
This follows by induction, using the $\mathfrak{su}(1,1)$ relations in Proposition \ref{su2prop}.
\end{proof}
In view of the explicit expressions for the CK-extension and the above lemma, one can expect the wavefunctions \eqref{Basis-1} to have an explicit expression in terms of hypergeometric functions. Consider the harmonic $\mathscr{Y}_{l_{n-1}\mathbf{l}}^{\epsilon_{n}\vec \epsilon}$ with $\mathbf{l}=\left(l_{n-2},\ldots,l_1\right)$ and $\vec \epsilon=\left(\epsilon_{n-1},\ldots,\epsilon_1\right)$ with $\epsilon_i \in \{0,1\}$. We can write this down as the CK-extension of a harmonic with one less variable.
\begin{align*}
\mathscr{Y}_{l_{n-1}\mathbf{l}}^{\epsilon_{n}\vec \epsilon}&=\mathbf{CK}_{x_n}^{\epsilon_n}\left( \rVert x_{\left[n-1\right]} \rVert^{2l_{n-1}} \mathscr{Y}_{\mathbf{l}}^{\vec \epsilon}\right) \\
&=\sum_{j=0}^\infty  \frac{(-1)^j x_n^{2j+\epsilon_n}\Delta_{\left[n-1\right]}^j}{2^{2j}j!\left(\mu_n+\frac{1}{2}+\epsilon_n\right)_{j}}\left( \rVert x_{\left[n-1\right]} \rVert^{2l_{n-1}} \mathscr{Y}_{\mathbf{l}}^{\vec \epsilon}\right) \\
&=\sum_{j=0}^{l_{n-1}}  \frac{(-1)^j (-l_{n-1})_j \left(-\deg{\mathscr{Y}_{\mathbf{l}}^{\vec \epsilon}}-l_{n-1}+1-\gamma_{\left[n-1\right]} \right)_j  x_n^{2j+\epsilon_n}}{j!\left(\mu_n+\frac{1}{2}+\epsilon_n\right)_{j}} \rVert x_{\left[n-1\right]} \rVert^{2(l_{n-1}-j)}\mathscr{Y}_{\mathbf{l}}^{\vec \epsilon}\\
&=x_n^{\epsilon_n}\rVert x_{\left[n-1\right]} \rVert^{2l_{n-1}}\mathscr{Y}_{\mathbf{l}}^{\vec \epsilon}\times \\
&\qquad \quad \sum_{j=0}^{l_{n-1}} \frac{ (l_{n-1})_j \left(-\deg{\mathscr{Y}_{\mathbf{l}}^{\vec \epsilon}}-l_{n-1}+1-\gamma_{\left[n-1\right]} \right)_j  }{j!\left(\mu_n+\frac{1}{2}+\epsilon_n\right)_{j} } \left( \frac{-x_n^2}{\rVert x_{[n-1]} \rVert^2} \right)^j \\
&=x_n^{\epsilon_n}\rVert x_{\left[n-1\right]} \rVert^{2l_{n-1}}\mathscr{Y}_{\mathbf{l}}^{\vec \epsilon} \times \\
& \qquad \quad {}_2F_1\left(-l_{n-1},-\deg{\mathscr{Y}_{\mathbf{l}}^{\vec \epsilon}}-l_{n-1}+1-\gamma_{\left[n-1\right]}; \mu_n+\frac{1}{2}+\epsilon_n,\frac{-x_n^2}{\rVert x_{[n-1]} \rVert^2}\right) \\
&=x_n^{\epsilon_n}\rVert x_{\left[n-1\right]} \rVert^{2l_{n-1}}\mathscr{Y}_{\mathbf{l}}^{\vec \epsilon}\times  \\
&\qquad \quad \frac{l_{n-1}!}{\left(\mu_{n}+\frac{1}{2}+\epsilon_n\right)_{l_{n-1}}} P_{l_{n-1}}^{\left(\mu_n-1/2+\epsilon_{n},-\deg(\mathscr{Y}_{l_{n-1}\mathbf{l}}^{\epsilon_{n}\vec \epsilon})-\gamma_{\left[n\right]}+1\right)}\left(\frac{2\rVert x_{[n]} \rVert^2}{\rVert x_{[n-1]} \rVert^2}-1\right).
\end{align*}
In the last step we replaced the hypergeometric function by a Jacobi polynomial. 
As one can see, a harmonic in $n$ variables can be written as the product of a harmonic in $n-1$ variables times a polynomial in $n$ variables. Repeating this step, one can write down an explicit expression for the original harmonic. Let $d_m=\epsilon_m+\sum_{i=1}^{m-1} 2l_i+\epsilon_i$, we have the following expression for the original harmonic:
\begin{align*} 
\mathscr{Y}_{l_{n-1}\mathbf{l}}^{\epsilon_{n}\vec \epsilon} =x_1^{\epsilon_1} \prod_{i=2}^n x_i^{\epsilon_i}&\left\|x_{\left[i-1\right]} \right\|^{2l_{i-1}}\times \\
&\frac{l_{i-1}!}{\left(\mu_{i}+\frac{1}{2}+\epsilon_i\right)_{l_{i-1}}} P_{l_{i-1}}^{\left(\mu_i-1/2+\epsilon_{i},-d_i-\gamma_{\left[i\right]}+1\right)}\left(\frac{2\rVert x_{[i]} \rVert^2}{\rVert x_{[i-1]} \rVert^2}-1\right).
\end{align*}

\begin{remark}
The same basis was obtained in a different way in \cite{Xu}. Our approach has the advantage that the action of the 
 Abelian subalgebra can be readily obtained, see Proposition \ref{ActionofAbelianAlgebra}.
\end{remark}

\subsection{The action of the Abelian subalgebra $\widehat{\mathcal{Y}}_n$}
It is of interest to determine the action of the Abelian subalgebra $\widehat{\mathcal{Y}}_n$, as defined in line (\ref{LabellingAbelianAlgebra}), on the Dunkl harmonics. The following proposition addresses that question.

\begin{proposition}\label{ActionofAbelianAlgebra}
Let $\mathscr{Y}_{\mathbf{l}}^{\vec \epsilon} \in \mathcal{H}_k^{\vec \beta}\left(\mathbb{R}^n\right)$ with $\beta_i=+$ if $\epsilon_i=0$ and $\beta_i=-$ if $\epsilon_i=1$ We have 
\[ C_{[m]}\mathscr{Y}_{\mathbf{l}}^{\vec \epsilon} =\frac{1}{4}\left( d_m+\gamma_{[m]}\right)\left(d_m+\gamma_{[m]}-2\right)\mathscr{Y}_{\mathbf{l}}^{\vec \epsilon}, \]
with $d_m=\epsilon_m+\sum_{i=1}^{m-1} 2l_i+\epsilon_i$ as before.
\end{proposition}
\begin{proof}
First observe that 
\[ C_{[n]}\mathscr{Y}_{\mathbf{l}}^{\vec \epsilon} =\frac{1}{4}\left( \deg{\mathscr{Y}_{\mathbf{l}}^{\vec \epsilon}}+\gamma_{[m]}\right)\left(\deg{\mathscr{Y}_{\mathbf{l}}^{\vec \epsilon}}+\gamma_{[m]}-2\right)\mathscr{Y}_{\mathbf{l}}^{\vec \epsilon}, \]
by using the definition of $C_A$ in our case, formula \eqref{Casimir}. Also observe that $d_n=\deg{\mathscr{Y}_{\mathbf{l}}^{\vec \epsilon}}=k$. 
We now consider $C_{[m]}$. A short calculation shows that $[C_{[m]},x_p]=[C_{[m]},\Delta_{[p-1]}]=0$ for $p>m$. This means that $\mathbf{CK}_{x_p}^{\epsilon_p}$ commutes with $C_{[m]}$ for $p>m$. This gives us
\begin{multline*}
C_{[m]}\mathscr{Y}_{\mathbf{l}}^{\vec \epsilon}(x_1,\ldots, x_{n})=
\\
\mathbf{CK}_{x_{n}}^{\epsilon_n}\bigg[\|{x_{[n-1]}\|^{2l_{n-1}}}\cdots C_{[m]}\mathbf{CK}_{x_{m}}^{\epsilon_{m}}\Big[\|x_{[m-1]}\|^{2l_{m-1}}\big[\cdots \mathbf{CK}_{x_1}^{\epsilon_{1}}[1]\big]\Big]\bigg] .
\end{multline*}
$C_{[m]}$ now acts on an element of $\mathcal{H}_{d_m}^{\vec \beta}\left(\mathbb{R}^m\right)$ and this action we have written down at the beginning of this proof. This gives the sought-after action.
\end{proof}
In light of this result, the algebra $\widehat{\mathcal{Y}}_n$ will be represented by diagonal matrices if it acts on the Dunkl harmonics. This algebra labels the basis vectors of this representation so we will call it a labelling Abelian algebra. In the next section the action of the full Racah algebra on the Dunkl harmonics will be  described.

\section{Connection coefficients and representations of $R(n)$}
In the previous section, a basis of Dunkl harmonics that diagonalizes the labelling Abelian subalgebra $\hat{\mathcal{Y}_n}$ has been constructed. There are many similar bases. Indeed, upon permuting the order in which the $\mathbf{CK}$-extensions are applied, one can find other bases of Dunkl-harmonics which will diagonalize other labelling Abelian subalgebras.  For example, consider the case where $n=4$ with $\mathbf{CK}_{x_4}^{\epsilon_4}$ and $\mathbf{CK}_{x_1}^{\epsilon_1}$ swapped. This leads to the wavefunctions
$$
\mathscr{Y'}_{\mathbf{l}}^{\vec \epsilon} := \mathbf{CK}_1^{\epsilon_1} \left( \|x_{234}\|^{2l_3} \mathbf{CK}_3^{\epsilon_3}\left(  \|x_{24}\|^{2l_2}
\mathbf{CK}_2^{\epsilon_2} \left(\|x_4\|^{2l_1}
\mathbf{CK}_4^{\epsilon_4}\left(1\right)  \right)\right)\right).
$$
The labelling Abelian subalgebra being diagonalized is $\langle C_{24},C_{234} \rangle $. As can be seen, the indices have been permuted according to the permutation in the $\mathbf{CK}$-extensions. To each basis constructed in this way corresponds a labelling Abelian subalgebra of the full Racah algebra $R(n)$. Conversely, for every labelling Abelian subalgebra, there is a basis that can be constructed with the $\mathbf{CK}$-extension.

Given a basis defined as the joint eigenvectors of a labelling Abelian algebra, one would wish to provide the action of the other generators of $R(n)$ on the basis elements in order to complete the construction of the representation of $R(n)$. We shall now indicate that this can be accomplished from the knowledge of the connection or overlap coefficients between bases associated to different labelling Abelian subalgebras. With this understood, we shall be able to complete the picture by giving an explicit characterization of the connection coefficients.

Suppose that a basis has been picked and that want to give the action of a certain generator $C$ on the elements of this fixed basis.
It is not difficult to convince oneself that every generator of the generalized Racah algebra $R(n)$ belongs in a labelling Abelian subalgebra. There is thus another basis, let us call it prime, in which $C$ is diagonal and Proposition 8 will give its eigenvalues. Now if the connection coefficients between the elements of the original basis and those of the prime basis are explicitly known, from simple linear algebra, it is clear that the action of $C$ in the original basis can be written down. This applies to any generator. Hence if all bases associated to labelling Abelian algebras can be connected to the one originally picked to construct a representation of $R(n)$, the action of all generators in that basis can be obtained with the help of the overlap coefficients. We shall now describe how these connection coefficients can be obtained.

\subsection{Rank one} We first give an overview of the rank one case, which has been treated in depth in the literature; see for example \cite{Genest&Vinet&Zhedanov-2014-2}. In this case, i.e. for the Racah algebra $R(3)$, the three labelling Abelian subalgebras $\langle C_{12} \rangle $, $\langle C_{23} \rangle$ and $\langle C_{13} \rangle$ have only one generator. Consider an irreducible representation $V$ of $R(3)$, and let  $ \langle \phi_k \rangle$ be a set of basis vectors for $V$ on which $C_{12}$ acts in a diagonal fashion. We shall use the parametrization of the representations that is provided by the realization of the basis elements as Dunkl harmonics as per section 5 in particular. From Proposition 8, the action of the central elements is thus taken to be:
\begin{equation*}
C_i\phi_k=\lambda_{i}\phi_k, \qquad C_{123}\phi_k=\lambda_{123}\phi_k, 
\end{equation*}
where
\begin{align*}
&\lambda_i = \frac{1}{4}\left(\epsilon_i+\mu_i+\frac{1}{2}\right)\left(\epsilon_i+\mu_i-\frac{3}{2}\right),\qquad i=1,2,3,
\\
&\lambda_{123} = \frac{1}{4}\left(d_3+\mu_1+\mu_2+\mu_3+\frac{3}{2}\right)\left(d_3+\mu_1+\mu_2+\mu_3-\frac{1}{2}\right).
\end{align*}
Also from Proposition \ref{ActionofAbelianAlgebra}, it is seen that the eigenvalues of $C_{12}$ on  $\langle \phi_k\rangle$ are of the form
\begin{align*}
C_{12} \phi_{k} = \omega_k \phi_{k},\qquad \omega_k = \frac{1}{4}(2k+\sigma)(2k+\sigma-2),\quad k=0,1,2,\ldots
\end{align*}
where $\sigma = \mu_1+\mu_2 + \epsilon_1 + \epsilon_2 + 1$. As shown in \cite{Genest&Vinet&Zhedanov-2014-2}, $C_{23}$ acts in a tridiagonal fashion on $\langle \phi_k \rangle$. Here we assume the basisvectors to be normalized. One has
\begin{align*}
C_{23} \phi_{k} = U_{k+1} \phi_{k+1} + B_{k} \phi_{k} + U_{k} \phi_{k-1},
\end{align*}
where
\begin{align*}
B_k = \frac{1}{2}\left(\lambda_{123}-\omega_k -\frac{(\lambda_2-\lambda_1)(\lambda_3-\lambda_{123})}{\omega_k}\right),
\end{align*}
and 
\begin{align*}
U_{k}^2 = \frac{k(k+\beta)(k+\alpha-\delta)(k+\alpha+\beta-\gamma)(k+\alpha)(k+\alpha+\beta)(k+\gamma)(k+\beta+\delta)}{(2k+\alpha+\beta-1)(2k+\alpha+\beta)^2(2k+\alpha+\beta+2)},
\end{align*}
with
\begin{align}
\label{para}
\begin{aligned}
&\alpha=\mu_1+\epsilon_1-\frac{1}{2},&\quad &\beta=\mu_2+\epsilon_2-\frac{1}{2}, 
\\
&\gamma=-\mu_3-\frac{d_3+1-\epsilon_1-\epsilon_2+\epsilon_3}{2},&\quad &\delta=-\mu_2-\frac{d_3+1-\epsilon_1+\epsilon_2-\epsilon_3}{2}.
\end{aligned}
\end{align}
Let us now consider a basis $ \langle \psi_s \rangle $ diagonalizing $\langle C_{23} \rangle$, that is $C_{23}\psi_s=\mu_s\psi_s$, where $\mu_s$ is the same as $\omega_s$ with the indices permuted according to the permutation $\pi = (123)$.  The connection coefficients $W_{sk}$ between the two bases are defined as
\[ \psi_s=\sum_k W_{sk}\phi_k. \]
Upon acting on both sides  with $C_{23}$ and using the action of $C_{23}$, one gets
\[ \mu_s\psi_s=\sum_k W_{sk}(U_{k+1}\phi_{k-1}+B_{k}\phi_k+U_{k}\phi_{k+1}). \]
Writing $\psi_{s}$ as a linear combination of the basis vectors $\phi_k$ and  introducing $P_k(\mu_s)W_{0s}=W_{ks}$ with $P_{0}(\mu_s)=1$, it follows that
\[ \mu_sP_{k}(\mu_s)=U_{k+1}P_{k+1}(\mu_s)+B_{k}P_{k}(\mu_s)+U_{k}P_{k-1}(\mu_s). \]
From this formula, it is seen that $P_k(\mu_s)$ are polynomials of degree $k$ in $\mu_s$. One can take $\hat{P}_k(x)=U_{1}\cdots U_{k} \,P_k(x)$, and the relation becomes
\[ x\hat{P}_k(x)=\hat{P}_{k+1}(x)+B_k\hat{P}_k(x)+U_k^2\hat{P}_{k-1}(x). \]
Introducing $\tilde{x}=x+\tau$ and $H_k(\tilde{x})=\hat{P}_k(x)$, one has
 \[ \tilde{x}H_k(\tilde{x})=H_{k+1}(\tilde{x})+(B_{k}+\tau)H_k(\tilde{x})+U_k^2H_{k-1}(\tilde{x}). \]
With $\tau=\frac{1}{4}(\gamma+\delta+1)(\gamma+\delta)$, this recurrence relation coincides with the one defining the Racah polynomials $R_n(x;\alpha,\beta,\gamma,\delta)$ with parameters $\alpha,\beta,\gamma,\delta$ as in \eqref{para}.  The interbasis expansion coefficients between the bases associated to $\langle C_{12}\rangle$ and $\langle C_{23}\rangle$ are thus expressed in terms of one-variable Racah polynomials. It is easily shown that this holds for any pair of Casimir operator $C_{ij}$ in the rank one Racah algebra. For convenience, the connection coefficients will be written simply as functions of $\lambda_{i}$, for $i=1,2,3$, and $\lambda_{123}$. We shall hence write
\[\psi_s=\sum_k W_{sk}(\lambda_{123},\lambda_1,\lambda_2,\lambda_3)\phi_k. \]

\subsection{Higher rank}
Let us now consider the connection coefficients in higher rank cases. For simplicity, the discussion below will pertain to the rank two case, $R(4)$. However, the analysis and the conclusions extend to any rank. First we  consider the overlap coefficients between two bases associated to two labelling Abelian subalgebras that differ by only one generator. An example of two such subalgebras for $R(4)$ is
\[ \langle C_{12},C_{123} \rangle, \quad \langle C_{12},C_{124} \rangle.\]
Let us look at the common eigenspaces of the generators appearing in both labelling Abelian subalgebras. In the above example these will be the eigenspaces of $C_{12}$. Since the generators that differ between the two algebras commute with the common ones, they will preserve the corresponding eigenspaces. In the example, these elements are $C_{123}$ and $C_{124}$. Now note that these elements are also the generators of the rank one algebra: $R^{12,3,4}(3)$ This follows from the discussion at the end of Section 4 by taking for the sets $K,L,M$: $K=\{1,2\}$, $L=\{3\}$, $M=\{4\}$. Manifestly the generators $C_{K\cup L}=C_{123}$, $C_{K\cup M}=C_{124}$, and $C_{L\cup M}=C_{34}$ of this rank one Racah algebra all commute with $C_{12}$, the common element of the two Abelian algebras. Hence, the subspaces spanned by the basis vectors with fixed eigenvalues of the common element $C_{12}$ will therefore support representations of $R(3)$. This allows to use the analysis for representations of the rank one Racah algebra which tells us that the connection coefficients are Racah polynomials. 

For example, take $\langle \phi_{j_1,j_2} \rangle$ and $\langle \psi_{j_1,j_2} \rangle$ to be the bases diagonalizing the labelling Abelian subalgebras in our example. The Abelian algebras act as follows:
\[ C_{12}\phi_{j_1,j_2}=\omega_{j_1}^{12}\phi_{j_1,j_2}, \qquad C_{123}\phi_{j_1,j_2}=\omega_{j_1+j_2}^{123}\phi_{j_1,j_2},  \]
and similarly for the other basis by replacing $\omega$ by $\mu$. We set $\omega_{j_1}^{12}=\mu_{j_1}^{12}$. The connection coefficients become:
\[ \phi_{j_1,j_2}=\sum_{k}W_{j_2k}(\lambda_{1234},\lambda_3,\omega_{j_1}^{12},\lambda_4)\psi_{j_1,k} \]
with $\lambda_{1234}$, $\lambda_3$ and $\lambda_4$ the scalar belonging to the central generators $C_{1234}$, $C_{3}$ and $C_4$. The observations made above can now be used to obtain the action of $C_{124}$ and $C_{34}$ on the basis vectors $\phi_{j_1, j_2}$. To find the action of other generators one must consider the relations of the basis $\{\phi_{j_1,j_2}\}$ with other subalgebra-type bases and check that the scheme allows to determine the action of all generators in this way. To address these points, it is useful to have the following picture in terms of a recoupling graph. Let every labelling Abelian subalgebra be represented by a vertex. Two vertices are connected by an edge if the corresponding bases only differ by one generator. Every edge thus represents a change of basis where the connection coefficients are univariate Racah polynomials.

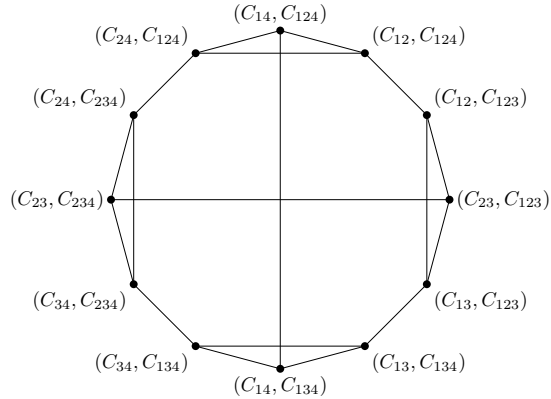
\begin{center}
	\scalebox{.75}{
\begin{tikzpicture}
\draw (30:3cm)--(60:3cm) ;
\draw (60:3cm)--(90:3cm);
\draw (90:3cm)--(120:3cm);
\draw (120:3cm)--(150:3cm);
\draw (150:3cm)--(180:3cm) ;
\draw (180:3cm)--(210:3cm);
\draw (210:3cm)--(240:3cm);
\draw (240:3cm)--(270:3cm);
\draw (270:3cm)--(300:3cm) ;
\draw (300:3cm)--(330:3cm);
\draw (330:3cm)--(0:3cm);
\draw (0:3cm)--(30:3cm);
\fill (30:3cm) circle [radius=2pt] node[anchor=south west] {$(C_{12},C_{123})$};
\fill (60:3cm) circle [radius=2pt] node[anchor=south west] {$(C_{12},C_{124})$};
\fill (90:3cm) circle [radius=2pt] node[anchor=south] {$(C_{14},C_{124})$};
\fill (120:3cm) circle [radius=2pt] node[anchor=south east] {$(C_{24},C_{124})$};
\fill (150:3cm) circle [radius=2pt] node[anchor=south east] {$(C_{24},C_{234})$};
\fill (180:3cm) circle [radius=2pt] node[anchor=east] {$(C_{23},C_{234})$};
\fill (210:3cm) circle [radius=2pt] node[anchor=north east] {$(C_{34},C_{234})$};
\fill (240:3cm) circle [radius=2pt] node[anchor=north east] {$(C_{34},C_{134})$};
\fill (270:3cm) circle [radius=2pt] node[anchor=north] {$(C_{14},C_{134})$};
\fill (300:3cm) circle [radius=2pt] node[anchor=north west] {$(C_{13},C_{134})$};
\fill (330:3cm) circle [radius=2pt] node[anchor=north west] {$(C_{13},C_{123})$};
\fill (0:3cm) circle [radius=2pt] node[anchor=west] {$(C_{23},C_{123})$};
\draw (180:3cm)--(0:3cm);
\draw (90:3cm)--(270:3cm);
\draw (60:3cm)--(120:3cm);
\draw (150:3cm)--(210:3cm);
\draw  (240:3cm) --(300:3cm) ;
\draw  (330:3cm)--(30:3cm) ;
\end{tikzpicture}}
\end{center}
Since the graph is connected, there is a path taking any given basis to all the others. The connection coefficients between any two bases are thus obtained by iterating along the edges of the path the procedure we described in the case of algebras with one element in common. The resulting overlap coefficients will be given by products of univariate Racah polynomials. Furthermore, since all generators are part of a labelling Abelian algebra (as is seen on the graph) and are diagonal in the basis attached to this algebra, the knowledge of the connection coefficients allows to obtain the action of all generators in a given fixed basis.

It is straightforward to see how these considerations extend from $R(4)$ to $R(n)$. We shall close this section by showing for arbitrary rank and using the recoupling graph that all bases associated to labelling Abelian algebras are related and that the representation in a selected subalgebra-type basis can be fully characterized.
\begin{proposition}
The recoupling graph of $R(n)$ is connected.
\end{proposition}
\begin{proof}
Any basis is a permutation of $\mathbf{CK}$-extensions of any other basis. Any permutation can be written as a succession of adjacent transpositions. It will suffice to prove that there exists a path between any pair of bases that differ by an adjacent transposition. 
Consider the following labelling Abelian subalgebra in $R(n)$:
\[ \langle C_{k_1k_2},C_{k_1k_2k_3},\ldots,C_{k_1\ldots k_{n-1}} \rangle. \]
To every labelling Abelian subalgebra we construct a unique list of ordered indices:
\[ [k_1,k_2,k_3,\ldots,k_n] \]
Going from one basis to another that differs by one element corresponds to flipping two adjacent indices in the list. This is an adjacent transposition. This concludes the proof. \end{proof}

As a closing remark we explain the connection with the tree method. Our basis is constructed adding one variable at a time. This corresponds to polynomials associated to trees in which each vertex has a leaf or is a leaf. In \cite{Scar} another set of polynomials was constructed using the tree method. To each tree they associated a set of multidimensional $q$-Hahn polynomials. The same article also constructed connection coefficients between these polynomials associated to different trees which turn out to be $q$-Racah polynomials. This is in agreement with our results. Even more so, they proved that any tree can be transformed into any other tree, which corresponds to the the proof of connectedness of our graph.

\section{Conclusion}

Summing up, we have introduced a higher rank generalization of the Racah algebra. It has been shown that this algebra is the symmetry algebra of both the $\mathbb{Z}_2^n$ Laplace-Dunkl equation and the generic superintegrable system on the $n$-sphere. Using a Cauchy-Kovalevskaia extension theorem, bases for the space of Dunkl harmonics were constructed. We explained how the interbasis expansion coefficients between these bases can be obtained, and how they allow to compute matrix elements for representations of the higher rank Racah algebra.
\section*{Acknowledgements}
\noindent
The research of HDB is supported by the Fund for Scientific Research-Flanders (FWO-V), project ``Construction of algebra realizations using Dirac-operators'', grant G.0116.13N. VXG holds a postdoctoral fellowship from the Natural Science and Engineering Research Council of Canada (NSERC). The research of LV is supported in part by NSERC. WVDV is grateful to the organizers of the conference `Dunkl operators, special functions and harmonic analysis', Paderborn, August 8 - 12, 2016, where these results were presented.


\end{document}